\documentclass{article}
\usepackage{amsmath,amssymb,amsthm}

\usepackage[utf8]{inputenc}
\usepackage[english]{babel}
\usepackage[active]{srcltx}
\usepackage{hyperref}
\usepackage[all]{xy}
\usepackage{amsfonts}
\usepackage{amsthm}
\usepackage{enumerate}
\usepackage{mathrsfs}
\usepackage{multicol}
\usepackage[all]{xy}
\usepackage{amsfonts}
\usepackage{latexsym,amsmath,amssymb}
\usepackage[usenames,dvipsnames]{color}

\newtheorem{theorem}{Theorem}

\newtheorem{prop}{Proposition}

\newtheorem{definition}{Definition}
\newtheorem{remark}{Remark}

\def\Real{\mathbb{R}}

\def\Cinfty{{\rm C}^\infty}

\newcommand{\derpar}[2]{\displaystyle\frac{\partial{#1}}{\partial{#2}}}

\newcommand{\restric}[2]{\left.#1\right|_{#2}}
\newcommand{\Lag}{\mathcal{L}}

\newcommand{\vf}{\mathfrak{X}}
\newcommand{\df}{\Omega}

\newcommand{\Tan}{\mathrm{T}}
\newcommand{\inn}{{\mathop{i}\nolimits}}
\newcommand{\Lie}{\mathop{\mathrm{L}}\nolimits}
\renewcommand{\d}{\mathrm{d}}
\def\hs{({\cal M},\Omega,\omega)}
\def\moment#1#2#3{{#1}_{#2}, \ldots, {#1}_{#3}}

\newcommand{\bal}{\begin{align*}}
\newcommand{\eal}{\end{align*}}
\def\beq{\begin{equation}}
\def\eeq{\end{equation}}
\def\bea{\begin{eqnarray}}
\def\eea{\end{eqnarray}}
\def\beann{\begin{eqnarray*}}
\def\eeann{\end{eqnarray*}}
\def\ben{\begin{enumerate}}
\def\een{\end{enumerate}}
\def\bit{\begin{itemize}}
\def\eit{\end{itemize}}
\def\ds{\(\displaystyle}

\title{\sc Variational principles and symmetries on fibered multisymplectic manifolds}
\author{ Jordi Gaset, Pedro D. Prieto-Mart\'inez, Narciso Rom\'an-Roy
 \\
\small Department of Mathematics.
   Ed. C-3, Campus Norte UPC\\
\small   C/ Jordi Girona 1. 08034 Barcelona. Spain.\\
\small     (gaset.jordi@gmail.com, peredaniel@ma4.upc.edu,
narciso.roman@upc.edu)
}

\begin{document}

\maketitle

\begin{abstract}
The standard techniques of variational calculus are
geometrically stated in the ambient of fiber bundles 
endowed with a (pre)multisymplectic structure.
Then, for the corresponding variational equations,
conserved quantities (or, what is equivalent, conservation laws), symmetries,
Cartan (Noether) symmetries, gauge symmetries
and different versions of Noether's theorem are studied
in this ambient.
In this way, this constitutes a general geometric framework for
all these topics that includes, as special cases,
first and higher order field theories and (non-autonomous) mechanics.
\end{abstract}


\noindent{\bf AMS s.\,c.\,(2010)}:  70S10, 70S05, 70H50,  49S05, 53D42, 55R10.\null

\section{Introduction}

As it is well known, the most of field equations of
first and higher-order classical field theories and mechanics 
are locally variational; that is, they can be obtained starting from a variational principle.
The phase spaces for all these theories have a similar geometric structure:
they are fiber bundles $\kappa\colon \mathcal{M}\to M$
over an orientable manifold $M$ (of dimension equal to $1$ for mechanical systems,
and greater than $1$ for field theories),
which are endowed with a multisymplectic or a pre-multisymplectic form (depending on the {\sl regularity} of the theory).

The aim of this review work is to state a generic geometric framework 
which allows us to include these variational principles 
for all these kinds of theories in a single formulation
(this is done in Section 2, after stablishing some previous geometric and mathematical background).
The variational equations, which are stated using multivector fields,
include the Euler-Lagrange as well as the Hamilton equations.
Then, we use this unified framework 
to study different kinds of symmetries for these equations,
their conserved quantities (i.e., conservation laws), 
and giving general versions of Noether's theorem
(in Section 3).

From this framework we can recover, as particular cases,
the variational principles and several topics
on the theory of symmetries and conserved quantities
for classical field theories and 
(non-autonomous) mechanics of first and higher order,
both in the Lagrangian and Hamiltonian formulations.
In particular,
let $\pi\colon E \longrightarrow M$ be a fiber bundle. Then,
If ${\cal M}\equiv J^k\pi$ (the $kth$-order jet bundle of $\pi$), 
and $\Omega\equiv\Omega_\Lag$
(the {\sl Poincar\'e-Cartan form} associated to a Lagrangian density $\Lag$),
we recover the classical {\sl Hamilton variational Principle}
and results on symmetries, conservation laws and Noether's theorem
for classical Lagrangian field theories of first order 
(if $k=1$) and higher-order (if $k>1$)
\cite{art:Aldaya_Azcarraga78_2,De-77,art:deLeon_etal2004,
EMR-96,FF-86,Gc-73,proc:Garcia_Munoz83,GMS-97,GS-73,
art:Kouranbaeva_Shkoller00,Krupka,KrupkaStepanova,
pere2,book:Saunders89}.
Taking the suitable {\sl multimomentum bundles} as ${\cal M}$, 
and the associated Hamiltonian counterparts of $\Omega_\Lag$, 
we recover the corresponding {\sl Hamilton-Jacobi variational Principle}
and symmetries for the Hamiltonian formalism 
of first and higher-order field theories
\cite{art:Aldaya_Azcarraga78_1,art:deLeon_etal2004,EMR-99b,art:Echeverria_DeLeon_Munoz_Roman07,HK-04}.
Finally, if in the above situations we take $M=\Real$,
we obtain the analogous results for the Lagrangian and Hamiltonian formalisms
of first and higher-order non-autonomous mechanics
\cite{Ar,LM-96,MS-98,pere1,SC-81}.

All the manifolds are real, second countable and $\Cinfty$. The maps and the structures are $\Cinfty$.  Sum over repeated indices is understood.


\section{Variational principle for multisymplectic systems}

\subsection{Multivector fields}

(See \cite{art:Echeverria_Munoz_Roman98} for details).
Let $\mathcal{M}$ be a $n$-dimensional differentiable manifold.

 \begin{definition}
Sections of $\Lambda^m\Tan \mathcal{M}$ are called 
{\rm $m$-multivector fields} in $\mathcal{M}$;
that is, they are the contravariant skew-symmetric tensors of order $m$ in $\mathcal{M}$.
The set of $m$-multivector fields in $\mathcal{M}$ is denoted as $\mathfrak{X}^m (\mathcal{M})$.
 \end{definition}

For every $\mathbf{X}\in\mathfrak{X}^m(\mathcal{M})$ and $p\in \mathcal{M}$, there exists a neighbourhood
 $U_p\subset \mathcal{M}$ and $X_1,\ldots ,X_r\in\mathfrak{X} (U_p)$ such that
$$\mathbf{X}\vert_{U_p}=\sum_{1\leq i_1<\ldots <i_m\leq r} f^{i_1\ldots i_m}X_{i_1}\wedge\ldots\wedge X_{i_m} \, ,
$$
with $f^{i_1\ldots i_m} \in C^\infty (U_p)$, $m \leqslant r\leqslant{\rm dim}\, \mathcal{M}$.

The classical operations with vector fields in differentiable manifolds 
can be extended to multivector fields.

 \begin{definition}
Let $\Omega\in\df^k({\cal M})$ be a differentiable $k$-form in ${\cal M}$
and let $\mathbf{X}\in\mathfrak{X}^m(\mathcal{M})$;  
the {\rm contraction} between ${\bf X}$ and $\Omega$
is defined as
 \beann
 \inn({\bf X})\Omega\mid_{U_p}&:=& \sum_{1\leq i_1<\ldots <i_m\leq
 r}f^{i_1\ldots i_m} \inn(X_1\wedge\ldots\wedge X_m)\Omega 
\\ &=&
 \sum_{1\leq i_1<\ldots <i_m\leq r}f^{i_1\ldots i_m} \inn
 (X_1)\ldots\inn (X_m)\Omega
 \eeann
 if $k\geq m$, and equal to zero if $k<m$.
For $1\leq j\leq k-1$, the $k$-form $\Omega$ is {\rm $j$-nondegenerate} if, for every $p\in E$ and
 ${\bf X}\in\vf^j({\cal M})$, we have that
$$
\inn({\bf X}_p)\Omega_p =0\  \Leftrightarrow \ {\bf X}_p=0 \ .
$$

The {\rm Lie derivative} with respect to ${\bf X}$ is defined as the graded bracket
 $$
 [\d , \inn ({\bf X})]=\d\inn ({\bf X})-(-1)^m\inn ({\bf X})\d:=\Lie ({\bf X})
 $$
and it is an operation of degree $m-1$.

If ${\bf Y}\in\vf^i({\cal M})$ and ${\bf X}\in\vf^j({\cal M})$,
another operation of degree $i+j-2$ is the
{\rm Schouten-Nijenhuis bracket} of ${\bf X},{\bf Y}$,
which is the bilinear assignment ${\bf Y},{\bf X}\mapsto [{\bf Y},{\bf X}]$, 
where $[{\bf Y},{\bf X}]$ is a $(i+j-1)$-multivector field obtained 
as the graded commutator of $\Lie ({\bf Y})$ and $\Lie ({\bf X})$;that is,
 $$
 [\Lie ({\bf Y}), \Lie ({\bf X})] := \Lie([{\bf Y},{\bf X}])\ .
 $$
 \end{definition}

The following properties hold:
for ${\bf X}$, ${\bf Y}$ and ${\bf Z}$, multivector fields of degrees $i,j,k$, respectively, we have that:
 \ben
 \item
 $[{\bf X},{\bf Y}] = -(-1)^{(i+1)(j+1)} [{\bf Y},{\bf X}]$.
 \item
 $[{\bf X},{\bf Y}\wedge{\bf  Z}] =
 [{\bf X},{\bf Y}]\wedge{\bf Z} + (-1)^{(i+1)j} {\bf Y}\wedge [{\bf X},{\bf Z}]$.
 \item
 $(-1)^{(i+1)(k+1)}[{\bf X},[{\bf Y},{\bf Z}] ] +
 (-1)^{(j+1)(i+1)}[{\bf Y},[{\bf Z},{\bf X}] ] +(-1)^{(k+1)(j+1)}[{\bf Z},[{\bf X},{\bf Y}]]
 = 0 $.
\item
For every $X\in\vf({\cal M})$,
$\inn([X,{\bf Y}]) \Omega =
 \Lie (X)\inn({\bf Y}) \Omega -\inn({\bf Y})\Lie(X)\Omega$.
 \een

 \begin{definition}
An $m$-multivector field $\mathbf{X} \in \mathfrak{X}^m(\mathcal{M})$
is said to be {\rm locally decomposable} if,
for every $p\in \mathcal{M}$, there exists an open neighbourhood  $U_p\subset \mathcal{M}$
and $X_1,\ldots ,X_m\in\mathfrak{X} (U_p)$ such that $\mathbf{X}\vert_{U_p}=X_1\wedge\ldots\wedge X_m$.
 \end{definition}

An $m$-dimensional distribution $D\subset T\mathcal{M}$
is \emph{locally associated} with a non-vanishing 
$m$-multivector field $\mathbf{X}\in\mathfrak{X}^m(\mathcal{M})$
if there exists a connected open set $U\subseteq \mathcal{M}$
such that $\mathbf{X}\vert_U$ is a section of $\Lambda^mD\vert_U$.
If $\mathbf{X},\mathbf{X}'\in\mathfrak{X}^m(\mathcal{M})$ are non-vanishing $m$-multivector fields 
locally associated with the same distribution $D$, 
on the same set $U$, then there exists a
non-vanishing function $f \in C^\infty (U)$ such that $\mathbf{X}'\vert_U=f\mathbf{X}\vert_U$. This
defines an equivalence relation in the set of non-vanishing $m$-multivector fields in
$\mathcal{M}$, whose equivalence classes are denoted by $\{ \mathbf{X}\}_U$. Therefore, there is a
one-to-one correspondence between the set of $m$-dimensional orientable distributions $D$ in
$T \mathcal{M}$ and the set of the equivalence classes $\{ \mathbf{X}\}_\mathcal{M}$ of
non-vanishing, locally decomposable $m$-multivector fields in $\mathcal{M}$.
If $\mathbf{X}\in\mathfrak{X}^m(\mathcal{M})$ is non-vanishing and locally decomposable, and
$U\subseteq \mathcal{M}$, we denote by 
$\mathcal{D}_U(\mathbf{X})$ 
(or simply $\mathcal{D}(\mathbf{X})$, if $U=\mathcal{M}$)
the distribution associated with the class $\{ \mathbf{X}\}_U$.

 \begin{definition}
A non-vanishing, locally decomposable multivector field 
$\mathbf{X}\in\mathfrak{X}^m(\mathcal{M})$ is 
{\rm integrable} or {\rm involutive} if  its associated distribution
$\mathcal{D}_U(\mathbf{X})$ is integrable or involutive.
 \end{definition}

Obviously, if
$\mathbf{X}\in\mathfrak{X}^m(\mathcal{M})$ is integrable or involutive, 
then so is every other in its equivalence class 
$\{ \mathbf{X}\}$, and all of them have the same integral manifolds.

We are especially interested in the case where 
$\kappa\colon \mathcal{M}\to M$ is a fiber bundle and $M$ is an $m$-dimensional
orientable manifold with volume form $\eta\in\df^m(M)$.

 \begin{definition}
 A multivector field $\mathbf{X}\in\mathfrak{X}^m(\mathcal{M})$ is 
 \emph{$\kappa$-transverse} if, for every $\beta\in\Omega^m(M)$ with
$\beta (\kappa(y))\not= 0$, at every point
$y\in \mathcal{M}$, we have that  $(i (\mathbf{X})(\kappa^*\beta))_y\not= 0$. 
If $\mathbf{X}\in\mathfrak{X}^m(\mathcal{M})$ is
integrable, then it is       
$\kappa$-transverse if, and only if, its integral manifolds are local sections of
$\kappa\colon \mathcal{M}\to M$.
In this case, if $\psi\colon U\subset M\to \mathcal{M}$ is a local
section with $\psi (x)=y$ and $\psi (U)$ is the integral manifold 
of $\mathbf{X}$ at $y$, then  $T_y({\rm Im}\,\psi) = \mathcal{D}_y(\mathbf{X})$
and $\psi$ is said to be an {\rm integral section} of ${\bf X}$.
 \end{definition}

Furthermore, there exists a unique $m$-multivector field
${\bf Y}_\eta\colon M\to\Lambda^m\Tan\,M$, such that $\inn({\bf Y}_\eta)\eta=1$; 
then the {\sl canonical prolongation} of a section 
$\psi\colon U\subset M\to \mathcal{M}$ to $\Lambda^m\Tan{\cal M}$ is
the section $\Lambda^m\psi\colon U\subset M\to\Lambda^m\Tan\mathcal{M}$ 
defined as $\Lambda^m\psi:=\Lambda^m\Tan\psi\circ{\bf Y}_\eta$; where 
$\Lambda^m\Tan\psi\colon\Lambda^m\Tan M\to\Lambda^m\Tan\mathcal{M}$ is the natural extension of $\psi$ 
to the corresponding multitangent bundles. Then,
$\psi$ is an integral section of ${\bf X}\in\vf^m({\cal M})$ if, and only if,
\beq
{\bf X}\circ\psi=\Lambda^m\psi \ .
\label{insec}
\eeq


 \subsection{(Pre)multisymplectic systems}

Let  $\kappa\colon{\cal M}\to M$ be a fibred manifold 
which in what follows is assumed to be a fibre bundle, where
$\dim\,M=m\geq 1$ and $\dim\,{\cal M}=n+m$, and $M$ is an
orientable manifold with volume form $\eta\in\df^m(M)$.
 We denote
$\omega=\kappa^*\eta$. We write $(U;x^\mu,y^j)$,
$\mu=1,\ldots ,m$, $j=1,\ldots ,n$, for local charts of
coordinates in ${\cal M}$ adapted to the fibred structure, and such that
$\omega=\d x^1\wedge\ldots\wedge\d x^m\equiv\d^mx$.
We denote by $\vf^{V(\kappa)}({\cal M})$ the set of
$\kappa$-vertical vector fields in ${\cal M}$
(which is locally generated by $\displaystyle \left\{\derpar{}{y^j}\right\}$).

 \begin{definition}
A form $\Omega\in\df^{m+1}({\cal M})$ ($m\geq 1$) is a
{\rm multisymplectic form} if it is closed and $1$-nondegenerate, that is, if the map 
$\flat_{\Omega}\colon\Tan {\cal M} \longrightarrow \Lambda^m\Tan^*{\cal M}$,
defined by $\flat_{\Omega}(x,v)=(x,\inn(v)\Omega_x)$,
for every $x\in{\cal M}$ and $v\in\Tan_x{\cal M}$, is injective. 
In this case, the system described by the triad 
$(F,\Omega,\omega)$ is called a {\rm multisymplectic system}.
 Otherwise, the form is said to be a
{\rm premultisymplectic form}, and the system is
 {\rm premultisymplectic}.
Finally, a multisymplectic form is {\rm exact} if there exist
$\Theta\in\df^m({\cal M})$ such that $\Omega=-\d\Theta$.
 \end{definition}

From now on, we will assume this last condition 
(this does not represent a loss of generality since,
by Poincar\'e Lemma, every closed form is locally exact).

Furthermore, if $m\geq 2$, we assume that the following condition holds:
$$
\inn(Z_1)\inn(Z_2)\inn(Z_3)\Omega=0\ , \mbox{\rm for every
$Z_1,Z_2,Z_3\in\vf^{{\rm V}(\kappa)}({\cal M})$}\ ,
$$
which is justified because this is the situation in
the Lagrangian and Hamiltonian formalism of field theories.
This condition means that, in a chart of adapted coordinates, we have that
\beq
\Omega\vert_U=\d F_j^\mu\wedge\d y^j\wedge\d^{m-1}x_\mu+\d E\wedge\d^m x \ ,
\label{omega}
\eeq
where $\displaystyle\d^{m-1}x_\mu=\inn\left(\derpar{}{x^\mu}\right)\d^mx$,
and
$F_j^\mu(x^\nu,y^i),E(x^\nu,y^i)\in\Cinfty(U)$.


\subsection{Generalized variational principle and field equations}
\label{varprincuf}


Let $\Gamma(\kappa)$ be the set of sections of $\kappa$. Consider the following functional
(where the convergence of the integral is assumed)
\begin{equation*}
\begin{array}{rcl}
{\cal F} \colon \Gamma(\kappa) & \longrightarrow & \Real \\
\psi & \longmapsto & \displaystyle \int_M \psi^*\Theta
\end{array} \, .
\end{equation*}

\begin{definition}[Generalized Variational Principle]
The {\rm generalized variational problem} for the (pre)multisymplectic system $\hs$
is the search for the critical (local) sections of
the functional ${\cal F}$ with respect to the variations
of $\psi$ given by $\psi_s = \sigma_s \circ \psi$,
where $\left\{ \sigma_s \right\}$ is a local one-parameter group of
any compact-supported $\kappa$-vertical vector field $Z$ in ${\cal M}$; that is,
\begin{equation*}
\restric{\frac{d}{ds}}{s=0}\int_M \psi_s^*\Theta = 0 \, .
\end{equation*}
\end{definition}


\begin{theorem}
\label{thm:EquivalenceVariationalSectionsUnified}
The following assertions on a section $\psi \in \Gamma(\kappa)$ are equivalent:
\begin{enumerate}
\item $\psi$ is a solution to the generalized variational problem.
\item $\psi$ is a section solution to the equation
\begin{equation}
\label{sect1}
\psi^*\inn(Y)\Omega= 0 \, , \quad \text{for every }Y \in \vf({\cal M}) \, .
\end{equation}
\item $\psi$ is a section solution to the equation
\begin{equation}
\label{sect2}
\inn(\Lambda^m\psi)(\Omega\circ\psi) = 0\ .
\end{equation}
\item $\psi$ is an integral section of a $m$-multivector field
contained in a class of $\kappa$-transverse and integrable (and hence locally decomposable) $m$-multivector fields,
$\left\{ {\bf X} \right\} \subset \vf^m({\cal M})$,
satisfying the equation
\begin{equation}
\label{vf}
\inn({\bf X})\Omega=0 \, .
\end{equation}
\end{enumerate}
\end{theorem}
\begin{proof}
(The proof follows the patterns in \cite{art:Echeverria_DeLeon_Munoz_Roman07}  and\cite{proc:Garcia_Munoz83}).

($1\,\Longleftrightarrow\, 2$)\quad
Let $Z \in \vf^{V(\kappa)}({\cal M})$ be a compact-supported vector field,
and $U \subset M$ an open set such that $\partial U$ is a $(m-1)$-dimensional
manifold and $\kappa({\rm supp}(Z)) \subset U$. Then
\beann
\restric{\frac{d}{ds}}{s=0} \int_M \psi^*_s\Theta
&=& \restric{\frac{d}{ds}}{s=0} \int_U \psi^*_s\Theta
= \restric{\frac{d}{ds}}{s=0} \int_U \psi^*\sigma_s^*\Theta
\\
&=& \int_U\psi^*\left( \lim_{t\to0} \frac{\sigma_s^*\Theta- \Theta}{t} \right)
= \int_U\psi^*\Lie(Z)\Theta \\
&=& \int_U \psi^*(\inn(Z)\d \Theta+ \d\inn(Z)\Theta)=
\int_U \psi^*(-\inn(Z)\Omega+ \d\inn(Z)\Theta)\\
&=& - \int_U \psi^*\inn(Z)\Omega+ \int_U \d(\psi^*\inn(Z)\Theta) \\
&=& - \int_U \psi^*\inn(Z)\Omega+ \int_{\partial U}\psi^*\inn(Z)\Theta
= - \int_U\psi^*\inn(Z)\Omega \, ,
\eeann
as a consequence of Stoke's theorem and the assumptions made
on the supports of the vertical vector fields. Thus,
we conclude
$$
\restric{\frac{d}{ds}}{s=0} \int_M \psi_s^*\Theta = 0 \quad
\Longleftrightarrow \quad
\psi^*\inn(Z)\Omega=0 \ ,
$$
for every compact-supported $Z \in \vf^{V(\kappa)}({\cal M})$.
However, since the compact-supported vector fields generate locally
the $\Cinfty({\cal M})$-module of vector fields in ${\cal M}$,
it follows that the last equality holds for every $\kappa$-vertical
vector field $Z$ in ${\cal M}$.
Now, recall that for every point $p \in {\rm Im}\psi$, we have a canonical
splitting of the tangent space of ${\cal M}$ at $p$ in a $\kappa$-vertical
subspace and a $\kappa$-horizontal subspace,
$$
\Tan_p{\cal M} = V_p(\kappa) \oplus \Tan_p({\rm Im}\psi) \, .
$$
Then, if $Y \in \vf({\cal M})$ we have
$$
Y_p = (Y_p - \Tan_p(\psi \circ \kappa)(Y_p)) + 
\Tan_p(\psi \circ \kappa)(Y_p) \equiv Y_p^V + Y_p^{\psi} \, ,
$$
with $Y_p^V \in V_p(\kappa)$ and $Y_p^{\psi} \in \Tan_p({\rm Im}\psi)$. Therefore
$$
\psi^*\inn(Y)\Omega= \psi^*\inn(Y^V)\Omega+ \psi^*\inn(Y^{\psi})\Omega= 
\psi^*\inn(Y^{\psi})\Omega \, ,
$$
since $\psi^*\inn(Y^V)\Omega = 0$, by the conclusion in the above
paragraph. Now, as $Y^{\psi}_p \in \Tan_p({\rm Im}\psi)$
for every $p\in {\rm Im}\psi$,
then the vector field $Y^{\psi}$ is tangent to ${\rm Im}\psi$,
and hence there exists a vector field
$X \in \vf(M)$ such that $X$ is $\psi$-related with $Y^{\psi}$;
that is, $\psi_*X = \restric{Y^{\psi}}{{\rm Im}\psi}$. Then
$\psi^*\inn(Y^{\psi})\Omega = \inn(X)\psi^*\Omega$. However, as
$\dim{\rm Im}\psi=\dim M =m$ and
$\Omega$ is an $(m+1)$-form, we obtain that
$\psi^*\Omega=0$ and hence $\psi^*\inn(Y^{\psi})\Omega = 0$.
Therefore, we conclude that the equation \eqref{sect1} holds.

Taking into account the reasoning of the first paragraph,
the converse is obvious
since the equation \eqref{sect1} holds for every $Y \in \vf({\cal M})$ and,
in particular, for every $Z \in\vf^{V(\kappa)}({\cal M})$.

($2\,\Longleftrightarrow\, 3$)\quad
In a chart of adapted coordinates $(U; x^\mu, y^j)$, for every $Y\in \vf({\cal M})$ and  for every $\psi\in\Gamma(\kappa)$ and $x\in M$,
we have that
$\displaystyle Y=f^\mu\derpar{}{x^\mu}+g^j\derpar{}{y^j}$,
$\psi(x)=(x^\mu,\psi^j(x))$ and 
$\displaystyle\Lambda^m\psi=\bigwedge_{\mu=1}^m\left(\derpar{}{x^\mu}+\derpar{\psi^j}{x^\mu}\derpar{}{y^j}\right)$.
Therefore, taking \eqref{omega} into account,
a simple calculation shows that 
equations \eqref{sect1} and \eqref{sect2} lead to the same expressions:
\beann
0 &=&\derpar{F^\mu_j}{x^\mu}\derpar{\psi^j}{x^\nu}+
\derpar{F^\mu_j}{y^i}\left(\derpar{\psi^i}{x^\nu}\derpar{\psi^j}{x^\mu}+\derpar{\psi^i}{x^\mu}\derpar{\psi^j}{x^\nu} \right)\ ,
\\
0 &=&\derpar{F^\mu_j}{x^\mu}+\derpar{F^\mu_i}{y^j}\derpar{\psi^i}{x^\mu}-
\derpar{F^\mu_j}{y^i}\derpar{\psi^i}{x^\mu}+\derpar{E}{y^j} \ .
\eeann
($3\,\Longleftrightarrow\, 4$)\quad
 If $\psi\colon U\subset M \to{\cal M}$ is a solution to \eqref{sect2} then,
for every $x\in U$ there exists a neigbourhood $U_x\subset U$ of $x$
such that $\psi(U_x)\subset \psi(U)$.
As $\psi\vert_{U_x}$ is an injective immersion
(since $\psi$ is a section and hence its image is an embedded submanifold),
the map $\Lambda^m(\psi\vert_{U_x})$ defines 
a locally decomposable $m$-multivector field ${\bf X}^x$
in $\psi(U_x)\subset M$, which is tangent to $\psi(U_x)$
and has $\psi\vert_{U_x}$ as an integral section in $U_x$. 
Thus, as a consequence of \eqref{insec}, if equation \eqref{sect2} holds for $\psi\vert_{U_x}$, 
then \eqref{vf} holds for ${\bf X}^x$, in $U_x$.

Conversely, if $\psi$ is an integral section of an $m$-multivector field
${\bf X}$ in $U\subset M$, then \eqref{insec} holds, and 
if \eqref{vf} holds for ${\bf X}$, then \eqref{sect2} holds for $\psi$, in $U$.
\end{proof}

\begin{remark}
The equation \eqref{vf}, with the $\kappa$-transverse condition, can be written
\beq
\inn({\bf X})\Omega=0 \quad ; \quad
\inn({\bf X})\omega\neq 0 \ .
 \label{fundeqs}
\eeq
Then, it is usual to fix the $\kappa$-transverse condition by taking
a representative in the class $\{{\bf X}\}$ such that
$$
\inn({\bf X})\omega=1\ .
$$
\end{remark}

As it is usual,
$\ker^m\Omega:=\{{\bf X}\in\vf^m({\cal M})\,\vert\,  \inn({\bf X})\Omega=0\}$. 
We denote by $\ker^m_\omega\Omega\subset\vf^m({\cal M})$
the set of $m$-multivector fields satisfying equations \eqref{fundeqs},
but not being locally decomposable necessarily. Then
$\ker^m_{\omega(ld)}\Omega\subset\vf^m({\cal M})$ 
and $\ker^m_{\omega(I)}\Omega\subset\vf^m({\cal M})$
denote the sets of $m$-multivector fields satisfying equations \eqref{fundeqs}
which are locally decomposable and integrable, respectively.
Obviously we have that
 \beq
\ker^m_{\omega(I)}\Omega\subset
\ker^m_{\omega(ld)}\Omega\subset\ker^m_\omega\Omega
 \subset\ker^m\Omega\ .
\label{chain}
 \eeq

\noindent {\bf Note}:
In general, if $({\cal M},\Omega,\omega)$ is a premultisymplectic system,
then $\kappa$-transverse and integrable $m$-multivector fields 
${\bf X}\in\vf^m({\cal M})$ which are  solutions to \eqref{fundeqs} 
could not exist. In the best of cases they exist only in some submanifold
$\j_{\cal S}\colon{\cal S}\hookrightarrow{\cal M}$ \cite{LMMMR-2005}.
In this case, in the sets of \eqref{chain} and in the following sections,
we have to consider only multivector fields and vector fields
which are tangent to ${\cal S}$.


\section{Symmetries and conservation laws for multisymplectic systems}

\subsection{Conserved quantities and conservation laws}
\label{symconlaws}

 Next we recover the idea of {\sl conservation law} or {\sl conserved quantity}, and state
 Noether's theorem for (pre)multisymplectic systems.
 In this sense, a part of our discussion is a generalization of
 the results obtained for non-autonomous mechanical systems
and field theories
(see \cite{LM-96,art:deLeon_etal2004,EMR-99b,SC-81}, and references therein).

 \begin{definition}
 A  {\rm conserved quantity}
 of a (pre)multisymplectic system $\hs$
 is a form $\xi\in\df^{m-1}({\cal M})$ such that
 $\Lie({\bf X})\xi=0$, for every
 ${\bf X}\in\ker^m_\omega\Omega$.
 \end{definition}

 Observe that, in this case,
 $\Lie({\bf X})\xi=(-1)^{m+1}\inn({\bf X})\d\xi$.

 \begin{prop}
\label{difxi}
 If $\xi\in\df^{m-1}({\cal M})$ is a
 first integral of a (pre)multisymplectic system $\hs$, and
 ${\bf X}\in\ker^m_{\omega(I)}\Omega$,
 then $\xi$ is closed on the integral submanifolds of ${\bf X}$;
 that is, if $j_S\colon S\hookrightarrow {\cal M}$
 is an integral submanifold of ${\bf X}$, then $\d j_S^*\xi=0$.
 \end{prop}
 \begin{proof}
 Let $X_1,\ldots ,X_m\in\vf ({\cal M})$ be independent vector fields
 tangent to the ($m$-dimensional) integral submanifold $S$. Then
 ${\bf X}=fX_1\wedge\ldots\wedge X_m$, for some $f\in\Cinfty ({\cal M})$.
 Therefore, as $\inn({\bf X})\d\xi=0$, we have
 $$
 j_S^*[\d\xi (X_1,\ldots ,X_m)]=
 j_S^*\inn(X_1\wedge\ldots\wedge X_m)\d\xi=0 \ .
 $$
 \end{proof}

 \begin{theorem}
 A form $\xi\in\df^{m-1}({\cal M})$
 is a conserved quantity of a (pre)multisymplectic system $\hs$ if, and only if,
 $\Lie({\bf Z})\xi=0$, for every
 ${\bf Z}\in\ker^m\Omega$.
 \label{fintchar}
 \end{theorem}
 \begin{proof}
 Let $\xi$ be a conserved quantity. 
If ${\bf X}_0\in\ker^m_\omega\Omega$
is a particular solution to the equations (\ref{fundeqs}) 
then
 $$
 \ker^m_\omega\Omega =
 \{f{\bf X}_0+\ker^m\,\Omega\cap \ker^m\omega \ ; \
 f\in\Cinfty({\cal M})\} \ .
 $$
 Therefore, for every ${\bf Z}\in\ker^m\Omega\cap \ker^m\omega$, we have that
 ${\bf Z}={\bf X}_1-{\bf X}_2$, with ${\bf X}_1,{\bf X}_2\in\ker^m_\omega\Omega$
 such that  $\inn({\bf X}_1)\omega=\inn({\bf X}_2)\omega$.
 Hence, if $\xi$ is a conserved quantity, we have that  $\Lie({\bf Z})\xi=0$.
 Furthermore, taking ${\bf X}_0\in\ker^m_\omega\Omega$
with $\inn({\bf X}_0)\omega=1$,
 for every ${\bf Z}\in\ker^m\Omega$
 we can write
 $ {\bf Z}=({\bf Z}-\inn({\bf Z})\omega{\bf X}_0)+
 \inn({\bf Z})\omega{\bf X}_0$
and it follows that
${\bf Z}-\inn({\bf Z})\omega{\bf X}_0\in \ker^m\Omega\cap\ker^m\omega$;
therefore $ \Lie({\bf Z}-\inn({\bf Z})\omega{\bf X}_0)\xi=0$ and thus
 $$
 \Lie({\bf Z})\xi=
 \Lie({\bf Z}-\inn({\bf Z})\omega{\bf X}_0)\xi+
 \Lie(\inn({\bf Z})\omega{\bf X}_0)\xi=
 (-1)^{m+1}\inn({\bf Z})\omega\inn({\bf X}_0)\d\xi=0 \ ,
 $$
 since $\d\inn({\bf X}_0)\xi=0$, because
 $\xi\in\df^{m-1}({\cal M})$.

 The converse is immediate.
 \end{proof}

Now, given $\xi\in\df^{m-1}({\cal M})$ and ${\bf X}\in\vf^m({\cal M})$,
for every integral submanifold $\psi\colon M\to{\cal M}$
of ${\bf X}$, we can construct the form 
$\psi^*\xi\in\df^{m-1}(M)$. Then, 
using the volume form $\eta\in\df^m(M)$, we can obtain a unique
$X_{\psi^*\xi}\in\vf(M)$ such that
$$
\inn(X_{\psi^*\xi})\eta=\psi^*\xi \ ,
$$
(in the standard terminology, $\psi^*\xi $ is the so-called
{\sl form of flux} associated with the vector field $X_{\psi^*\xi}$). Then:

\begin{prop}
If ${\rm div}X_{\psi^*\xi}$ denotes the divergence of $X_{\psi^*\xi}$, we have that
$$
({\rm div}X_{\psi^*\xi})\,\eta=\d{\psi^*\xi} \ .
$$
\end{prop}
\begin{proof}
In fact, \
$\d{\psi^*\xi}=\inn(X_{\psi^*\xi})\eta=
\Lie(X_{\psi^*\xi})\eta=({\rm div}X_{\psi^*\xi})\,\eta$\ .
\end{proof}

As a consequence of Proposition \ref{difxi},
this result allows to associate 
a {\sl conservation law} in $M$ to every conserved quantity in ${\cal M}$. In fact:

\begin{prop}
$\xi\in\df^{m-1}({\cal M})$ is a conserved quantity if, 
and only if, ${\rm div}X_{\psi^*\xi}=0$,
for every integral submanifold $\psi\colon M\to{\cal M}$
of ${\bf X}$.
Therefore, by {\sl Stokes theorem}, in every bounded domain $U\subset M$,
we have
$$
\int_{\partial U}{\psi^*\xi}=\int_U ({\rm div}X_{\psi^*\xi})\,\eta=\int_U\d{\psi^*\xi}=0 \ .
$$
The form $\psi^*\xi$ is called the {\rm current} associated with the conserved quantity $\xi$.
\end{prop}

\subsection{Symmetries}
\label{syms}

 \begin{definition}
\ben
\item
A {\rm symmetry} of a (pre)multisymplectic system 
$\hs$ is a diffeomorphism $\Phi\colon{\cal M}
\to{\cal M}$
such that $\Phi_*(\ker^m\Omega)\subset\ker^m\Omega$.
%
\item
 An {\rm infinitesimal symmetry}
 of a (pre)multisymplectic system $\hs$ is a vector field $Y\in\vf ({\cal M})$
 whose local flows are local symmetries; that is,
if $F_t$ is a local flow of $Y$, then
$F_{t*}(\ker^m\Omega)\subset\ker^m\Omega$,
in the corresponding open sets.
\een
 \end{definition}

Another characterizacion of infinitesimal symmetries is the following:

 \begin{theorem}
 Let $\hs$ be a (pre)multisymplectic system, $Y\in\vf({\cal M})$.
 Then $Y$ is an infinitesimal symmetry if, and only if,
$$
[Y,\ker^m\Omega]\subset\ker^m\Omega \ .
$$
 \label{previo0}
 \end{theorem}
\begin{proof}
 As $\ker^m\Omega$ is locally finite-generated, 
we can take a local basis $\moment{{\bf Z}}{1}{r}$ of $\ker^m\,\Omega$.
 Then, if $[Y,\ker^m\Omega]\subset\ker^m\Omega$, 
the assertion  is equivalent to proving that, if $F_t$ is a local flow of $Y$, then
 $[Y,{\bf Z}_i]=f_i^j{\bf Z}_j$ if, and only if, $F_{t*}{\bf Z}_i=g_i^j{\bf Z}_j$
 (for every $i=1,\ldots ,r$), where $g_i^j$ are
 functions defined on the corresponding open set,
 also depending on $t$.

First, it is clear that, if $F_{t*}{\bf Z}_i=g_i^j{\bf Z}_j$, then 
$[Y,{\bf Z}_i]=f_i^j{\bf Z}_j$.

 For the converse,
 suppose that $[Y,{\bf Z}_i]=f_i^j{\bf Z}_j$, and consider and extended local basis to the whole $\vf^m(\cal M)$: $\{\bf {Z}_1,\dots,\bf Z_r,\bf Z'_{r+1},\dots,\bf Z'_c\}$, where $c$ is the dimension of $\vf^m(\cal M)$. Remember that
 \ds\frac{\d}{\d t}\Big\vert_{t=s}F_{t*}{\bf Z}_i=F_{s*}[Y,{\bf Z}_i]\) .  Hence, on the one hand we obtain
 \beann
 F_{s*}[Y,{\bf Z}_i] &=&
 F_{s*}(f_i^j{\bf Z}_j)=(F_s^{-1})^*f_i^jF_{s*}{\bf Z}_j
\\ &=&
 (F_s^{-1})^*f_i^j(g_j^k{\bf Z}_k)+ \sum_{k=r+1}^c(F_s^{-1})^*f_i^j(g_j^k{\bf Z'}_k)  \ ,
 \eeann
 and on the other hand, we have that
 \beann
 \frac{\d}{\d t}\Big\vert_{t=s}F_{t*}{\bf Z}_i&=&
 \frac{\d}{\d t}\Big\vert_{t=s}g_i^k{\bf Z}_k+ \sum_{k=r+1}^c\frac{\d}{\d t}\Big\vert_{t=s}g_i^k{\bf Z'}_k
\\ &=&
 \frac{\d g_i^k}{\d t}\Big\vert_{t=s}{\bf Z}_k +
 \sum_{k=r+1}^c\frac{\d g_i^k}{\d t}\Big\vert_{t=s}{\bf Z'}_k
 \ .
\eeann
 Therefore, comparing these expressions, we conclude that
 $\displaystyle\frac{\d g_i^k}{\d t}=(F_t^{-1})^*f_i^jg_j^k$, for $k=1,\dots,c$.
 This is a system of ordinary linear differential equations
 for the functions $g_i^k$. With the initial condition $g_i^k(0)=\delta_i^k$ for $k\leq r$ and $g_i^k(0)=0$ for $k>r$,
 has a unique solution, defined for every $t$ on the domain of
 $F_t$. Then, taking this solution, 
we have proved the existence of functions
 $g_i^j$ such that $F_{t*}{\bf Z}_i=g_i^j{\bf Z}_j$,
and the result holds.
 \end{proof}

 Bearing in mind the properties of multivector fields
 we obtain the basic properties:
 \begin{itemize}
  \item
 If $Y_1,Y_2\in\vf ({\cal M})$ are infinitesimal symmetries,
 then so is $[Y_1,Y_2]$.
  \item
 If $Y\in\vf ({\cal M})$ is an infinitesimal symmetry and $\Omega$
is a premultisymplectic form, for every $Z\in\ker\Omega$,
then $Y+Z$ is also an infinitesimal symmetry.
 \end{itemize}

The classical interpretation that a symmetry of a system of
differential equations transforms solutions into solutions
is recovered from the following result:

 \begin{theorem}
 Let $\Phi\in{\rm Diff}({\cal M})$ be a symmetry of a (pre)multisymplectic system.
 \ben
 \item
 If ${\bf X}\in\ker^m\Omega$ is an integrable
 multivector field, then $\Phi$ transforms integral submanifolds
 of ${\bf X}$ into integral submanifolds of $\Phi_*{\bf X}$.
 \item
 In particular, if \,$\Phi\in{\rm Diff}({\cal M})$ restricts to a diffeormorphism 
$\varphi\colon M\to M$ (that is, $\varphi\circ\kappa=\kappa\circ\Phi$),
then, for every ${\bf X}\in\ker^m_{\omega(I)}\Omega$,
$\Phi$ transforms integral submanifolds of ${\bf X}$
into integral submanifolds of $\Phi_*{\bf X}$, and hence
$\Phi_*{\bf X}\in\ker^m_{\omega(I)}\Omega$.
 \een
 \label{gsymsol}
 \end{theorem}
 \begin{proof}
 \ben
 \item
 Let $X_1,\ldots ,X_m\in\vf ({\cal M})$ be vector fields
 locally expanding the involutive distribution associated with
 ${\bf X}$. Then $\Phi_*X_1,\ldots ,\Phi_*X_m$ generate another
 distribution which is also involutive, and, hence, is associated
 with a class of locally decomposable multivector fields
 whose representative is just $\Phi_*{\bf X}$, by construction.
 The assertion about the integral submanifolds is then immediate.
 \item
As $\Phi\in{\rm Diff}({\cal M})$ restricts to a diffeomorphism 
$\varphi$ in $M$ such that $\varphi\circ\kappa=\kappa\circ\Phi$ then,
for every $\psi\colon M\to {\cal M}$, integral section of
 ${\bf X}$, we can define $\psi_M\colon M\to {\cal M}$
 as $\Phi\circ\psi=\psi_M\circ \varphi$,
 which is also a section of $\kappa$ because
 $$
 \kappa\circ\psi_M =
 \kappa\circ\Phi\circ\psi\circ (\varphi)^{-1}=
 \varphi\circ\kappa\circ\psi\circ (\varphi)^{-1}=
 \varphi\circ (\varphi)^{-1}= {\rm Id}_M \ ,
 $$
 since $\kappa\circ\psi={\rm Id}_M$.
 Then, by construction, ${\rm Im}\,\psi_M=\Phi({\rm Im}\,\psi)$ is
 an integral submanifold of $\Phi_*{\bf X}$,
 and as is a section of $\kappa$, it is
 $\kappa$-transverse. Hence $\Phi_*{\bf X}$
 (which belongs to $\ker^m\Omega$,
 by Theorem \ref{previo0}) is integrable (then locally decomposable),
 and as its integral submanifolds are sections
 of $\kappa$, then $\Phi_*{\bf X}$ is $\kappa$-transverse, and
 thus $\Phi_*{\bf X}\in\ker^m_{\omega(I)}\Omega$. 
 \een
\end{proof}

From this result we obtain as an immediate corollary the following:

 \begin{theorem}
 Let $Y\in\vf ({\cal M})$ be an infinitesimal symmetry of a (pre)multisymplectic system
$\hs$, and $F_t$ a local flow of $Y$.
 \ben
 \item
 If ${\bf X}\in\ker^m\Omega$ is an integrable
 multivector field, then $F_t$ transforms integral submanifolds
 of ${\bf X}$ into integral submanifolds of $F_{t*}{\bf X}$.
 \item
In particular, if $Y\in\vf ({\cal M})$ is $\kappa$-projectable
(this means that there exists $Z\in\vf (M)$ such that
the local flows of $Z$ and $Y$ are $\kappa$-related),
then, for every ${\bf X}\in\ker^m_{\omega(I)}\Omega$,
$F_t$ transforms integral submanifolds of ${\bf X}$
into integral submanifolds of $F_{t*}{\bf X}$, and hence
$F_{t*}{\bf X}\in\ker^m_{\omega(I)}\Omega$.
 \een
 \label{gsymsol2}
 \end{theorem}

Symmetries allows us to obtain new conserved quantities from another one:

 \begin{prop}
\ben
\item
If $\Phi\in{\rm Diff}({\cal M}$ is a symmetry and $\xi\in\df^{m-1}({\cal M})$
is a conserved quantity of a (pre)multisymplectic system $\hs$, 
then $\Phi^*\xi$ is also a conserved quantity.
\item
 If $Y\in\vf ({\cal M})$, is an infinitesimal symmetry and $\xi\in\df^{m-1}({\cal M})$
is a conserved quantity of a (pre)multisymplectic system $\hs$, 
then $\Lie(Y)\xi$ is also a conserved quantity.
\een
 \end{prop}
 \begin{proof}
 For every ${\bf X}\in\ker^m\Omega$, we have that:
\ben
\item
As $\Phi_*{\bf X}\in\ker^m\Omega$, we obtain:
$$
\Lie({\bf X})(\Phi^*\xi)=\Phi^*\Lie(\Phi_*{\bf X})\Omega=0 \ .
$$
\item
As $[{\bf X},Y]\in\ker^m\Omega$,
as a consequence of Theorem \ref{fintchar} we get
 $$
 \Lie({\bf X})\Lie(Y)\xi =
 \Lie([{\bf X},Y])\xi+\Lie(Y)\Lie({\bf X})\xi=
 \Lie([{\bf X},Y])\xi=0 \ .
 $$
\een
 \end{proof}

 \subsection{Cartan symmetries. Noether's theorem}
\label{noetherth}

Now we introduce the concept that generalizes the notion
of {\sl Cartan (Noether) symmetry} for non-autonomous mechanical systems \cite{LM-96,SC-81}.

 \begin{definition}
\ben
\item
A {\rm Cartan symmetry} of a (pre)multisymplectic system 
$\hs$ is a diffeomorphism $\Phi\colon{\cal M}
\to{\cal M}$
such that, $\Phi^*\Omega=\Omega$.
\item
An {\rm infinitesimal Cartan symmetry}
 of a (pre)multisymplectic system $\hs$ is a vector field $Y\in\vf ({\cal M})$
 satisfying that $\Lie(Y)\Omega=0$.
\een
 \label{CNsym}
 \end{definition}

\noindent {\bf Remarks}:
 \bit
 \item
 It is immediate to prove that, if $Y_1,Y_2\in\vf ({\cal M})$
 are infinitesimal Cartan symmetries, then so is $[Y_1,Y_2]$.
 \item
The condition $\Lie(Y)\Omega=0$
 is equivalent to demanding that $\inn(Y)\Omega$
 is a closed $m$-form in ${\cal M}$. Therefore, for every
 $p\in {\cal M}$, there exists an open neighborhood $U_p\ni p$,
 and $\xi_Y\in\df^{m-1}(U_p)$, such that
 $\inn(Y)\Omega=\d\xi_Y$ (on $U_p$).
 Thus, an infinitesimal Cartan symmetry of a (pre)multisymplectic system is
 just a {\sl locally Hamiltonian vector field}
 for the multisymplectic form $\Omega$,
 and $\xi_Y$ is the corresponding {\sl local Hamiltonian form},
 which is unique, up to a closed $(m-1)$-form.
 \eit

 \begin{prop}
\ben
\item
Every Cartan symmetry of a (pre)multisymplectic system 
$\hs$ is a symmetry.
 \item
 Every infinitesimal Cartan symmetry of a (pre)multisymplectic system $\hs$
 is an infinitesimal symmetry.
\een
 \end{prop}
 \begin{proof}
 For every ${\bf X}\in\ker^m\Omega$, we have that:
 \ben
\item
If $\Phi\in{\rm Diff}({\cal M})$ is a Cartan symmetry then
\beann
\Phi^*\inn(\Phi_*{\bf X})\Omega=\inn({\bf X})(\Phi^*\Omega)=
\inn({\bf X})\Omega=0
\ &\Longleftrightarrow& \
\inn(\Phi_*{\bf X})\Omega=0
\\  &\Longleftrightarrow&\
\Phi_*{\bf X}\in\ker^m\Omega \ .
\eeann
 \item
If $Y\in\vf({\cal M})$ is an infinitesimal Cartan symmetry, then
$$
 \inn([Y,{\bf X}])\Omega= \Lie(Y)\inn({\bf X})\Omega-\inn({\bf X})\Lie(Y)\Omega=0 \ \Longleftrightarrow \ [Y,{\bf X}]\subset\ker^m\Omega \ .
$$
(Also, if $Y\in\vf({\cal M})$ is an infinitesimal Cartan symmetry,
by definition, its local flows are local Cartan symmetries,
then the result is a consequence of the above item).
\een
 \end{proof}

Then, the classical {\sl Noether's theorem}
can be generalized as follows:

 \begin{theorem}
 {\rm (Noether):}
 If $Y\in\vf ({\cal M})$ is an infinitesimal Cartan symmetry of a  (pre)multisymplectic
 system $\hs$, with $\inn(Y)\Omega=\d\xi_Y$. Then,
 for every ${\bf X}\in\ker^m_\omega\Omega$
(and hence for every ${\bf X}\in\ker^m_{\omega(I)}\Omega$), we have that
 $$
 \Lie({\bf X})\xi_Y=0 \ ;
 $$
 that is, any Hamiltonian $(m-1)$-form
 $\xi_Y$ associated with $Y$ is a conserved quantity of $\hs$.
(It is usually called a {\rm Noether current}, in this context).
 \label{Nth}
 \end{theorem}
 \begin{proof}
 If $Y\in\vf ({\cal M})$ is a Cartan symmetry then
 $$
 \Lie({\bf X})\xi_Y=\d\inn({\bf X})\xi_Y
 -(-1)^m\inn({\bf X})\d\xi_Y=
 -(-1)^m\inn({\bf X})\inn(Y)\Omega=
 -\inn(Y)\inn({\bf X})\Omega=0 \ .
 $$
 \end{proof}

To our knowledge,
 given a conserved quantity of a (pre)multisymplectic system,
 there is no a straightforward way of associating to it
 an infinitesimal Cartan symmetry $Y$ since, given a $(m-1)$-form 
$\xi$, the existence of a solution to the equation
 $\inn(Y)\Omega=\d\xi$ is not assured
 (even in the case $\Omega$ being 1-nondegenerate).
 Hence, in general, the {\sl converse Noether theorem} cannot be stated
 for (pre)multisymplectic systems.

Finally, as a particular case, we have:

 \begin{prop}
 Let $Y\in\vf ({\cal M})$ be an infinitesimal Cartan symmetry
 of an exact (pre)multisymplectic system $\hs$ (with $\Omega=-\d\Theta$). Therefore:
 \ben
 \item
 $\Lie(Y)\Theta$ is a closed form,
 hence, in an open set $U\subset {\cal M}$, there exist
 $\zeta_Y\in\df^{m-1}(U)$ such that $\Lie(Y)\Theta=\d\zeta_Y$.
 \item
 If $\inn(Y)\Omega=\d\xi_Y$,
 in an open set $U\subset {\cal M}$, then
 $$
 \Lie(Y)\Theta=\d (\inn(Y)\Theta-\xi_Y)=\d\zeta_Y
 \quad \mbox{\rm (in $U$)} \ ,
 $$
and hence $\xi_Y=\inn(Y)\Theta-\zeta_Y$ (up to a closed $(m-1)$-form).

As a particular case, if $\Lie (Y)\Theta=0$, we can take
 $\xi_Y=\inn(Y)\Theta$, and
$Y$ is said to be an {\rm exact infinitesimal Cartan symmetry}.
 \een
 \label{xizeta}
 \end{prop}
 \begin{proof}
 \ben
 \item
 The first item is immediate since
 $\d\Lie(Y)\Theta=\Lie(Y)\d\Theta=0$.
 \item
 For the second item we have
 $$
 \Lie(Y)\Theta=
 \d\inn(Y)\Theta+\inn(Y)\d\Theta=
 \d\inn(Y)\Theta-\inn(Y)\Omega=
 \d (\inn(Y)\Theta-\xi_Y) \ .
 $$
 Hence we can write $\xi_Y=\inn(Y)\Theta-\zeta_Y$
 (up to a closed $(m-1)$-form).
 \een
 \end{proof}

In the case that $\ker\Omega:=\{ Y\in\vf({\cal M})\,\vert\, \inn(Y)\Omega=0\}\not=\{ 0\}$, these vector fields are Cartan symmetries. Then:

\begin{definition}
Let $\hs$ be a premultisymplectic system such that 
the equations \eqref{fundeqs} have solutions on ${\cal M}$.
Then $Y\in\vf({\cal M})$ is a {\rm gauge symmetry} of $\hs$  if \,$Y\in\ker\Omega$.
\end{definition}

 \subsection{Higher-order Cartan symmetries.\\
Generalized Noether's theorem}
\label{gennoetherth}

 Noether's theorem associates conserved quantites to Cartan
 symmetries. But there are symmetries which are not of Cartan type.
 Different attempts have been made to extend Noether's theorem in order
 to obtain the corresponding conservation laws for these kinds of symmetries.
Next we present a generalization of
theorem \ref{Nth}, which is based in the approach of
\cite{SC-81} for mechanical systems.

 \begin{definition}
 An {\rm infinitesimal Cartan symmetry of order $n$}
 of a  (pre) multisymplectic system $\hs$ is a vector field $Y\in\vf ({\cal M})$
 satisfying that:
 \ben
 \item
 $Y$ is a symmetry of the (pre)multisymplectic system $\hs$.
 \item
 $\Lie^n(Y)\Omega=0$, but
 $\Lie^k(Y)\Omega\not= 0$, for $k<n$.
 \een
 \label{CNsymn}
 \end{definition}

Cartan symmetries of order $n>1$ are not
 necessarily Hamiltonian vector fields for the (pre)multisymplectic form
 $\Omega$. Nevertheless we have that:

 \begin{prop}
 If $Y\in\vf ({\cal M})$ is an infinitesimal Cartan symmetry of order $n$
 of a (pre)multisymplectic system $\hs$, then the form
 $\Lie^{n-1}(Y)\inn(Y)\Omega\in\df^m({\cal M})$
 is closed.
 \end{prop}
 \begin{proof}
 In fact, from the definition \ref{CNsymn} we obtain
 $$
 0=\Lie^n(Y)\Omega=
 \Lie^{n-1}(Y)\Lie(Y)\Omega=
 \Lie^{n-1}(Y)\d\inn(Y)\Omega=
 \d\Lie^{n-1}(Y)\inn(Y)\Omega \ .
 $$
 \end{proof}

 This condition is equivalent to demanding that, for every
 $p\in{\cal M}$, there exists an open neighborhood $U_p\ni p$,
 and $\xi_Y\in\df^{m-1}(U_p)$, such that
 $\Lie^{n-1}(Y)\inn(Y)\Omega=\d\xi_Y$ (on $U_p$).
Then, theorem \ref{Nth} can be generalized 
as follows:

 \begin{theorem}
 If $Y\in\vf ({\cal M})$ is an infinitesimal Cartan symmetry of order $n$
 of a (pre)multisymplectic system $\hs$,
 with $\Lie^{n-1}(Y)\inn(Y)\Omega=\d\xi_Y$. Then,
 for every ${\bf X}\in\ker^m_\omega\Omega$
(and hence for every ${\bf X}\in\ker^m_{\omega(I)}\Omega$), we have that
 $$
 \Lie({\bf X})\xi_Y=0
 $$
 that is, the $(m-1)$-form
 $\xi_Y$ associated with $Y$ is a conserved quantity.
 \label{Nthgen}
 \end{theorem}
 \begin{proof}
 If $Y\in\vf (J^{1*}E)$ is an infinitesimal Cartan symmetry of order $n$ then
 it is a symmetry, and then
 $[Y,{\bf X}]={\bf Z}\in\ker\,\Omega$.
 Therefore
 \beann
 \Lie({\bf X})\xi_Y &=&
 (-1)^{m+1}\inn({\bf X})\d\xi_Y=
 (-1)^{m+1}\inn({\bf X})\Lie^{n-1}(Y)\inn(Y)\Omega
 \\ &=&
 (-1)^{m+1}\inn({\bf X})\Lie(Y)\Lie^{n-2}(Y)\inn(Y)\Omega
  \\ &=&
 (-1)^{m+1}(\Lie(Y)\inn({\bf X})\Lie^{n-2}(Y)\inn(Y)\Omega
 -\inn([Y,{\bf X}])\Lie^{n-2}(Y)\inn(Y)\Omega)
  \\ &=&
 (-1)^{m+1}((\Lie(Y)\inn({\bf X})-\inn({\bf Z})\Lie^{n-2}(Y))\inn(Y)\Omega) \ ,
 \eeann
 and repeating the reasoning $n-2$ times we arrive at the
 result
 $$
 \Lie({\bf X})\xi_Y=
 (-1)^{m+1}((\Lie(Y)\inn({\bf X})-\inn({\bf Z}))^{n-1}
 \inn(Y)\Omega)=0 \ ,
 $$
 since $\inn({\bf X})\inn(Y)\Omega=0$
 and $\inn({\bf Z})\inn(Y)\Omega=0$.
 \end{proof}


 \begin{prop}
 Let $Y\in\vf ({\cal M})$ be an infinitesimal Cartan symmetry of order $n$
 of an exact (pre)multisymplectic system $\hs$. Therefore:
 \ben
 \item
 $\Lie^n(Y)\Theta$ is a closed form,
 hence, in an open set $U\subset {\cal M}$, there exist
 $\zeta_Y\in\df^{m-1}(U)$ such that  $\Lie^n(Y)\Theta=\d\zeta_Y$.
 \item
 If $\Lie^{n-1}(Y)\inn(Y)\Omega=\d\xi_Y$,
 in an open set $U\subset J^{1*}E$, then
 $$
 \Lie^n(Y)\Theta=
 \d (\Lie^{n-1}(Y)\inn(Y)\Theta-\xi_Y)=\d\zeta_Y
 \quad \mbox{\rm (in $U$)} \ .
 $$
 \een
\label{newcoro}
 \end{prop}
 \begin{proof}
 \ben
 \item
 The first item is immediate since
 $\d\Lie^n(Y)\Theta=\Lie^n(Y)\d\Theta=0$.
 \item
 For the second item we have
 \beann
 \Lie^n(Y)\Theta &=&
 \Lie^{n-1}(Y)\Lie(Y)\Theta=
 \Lie^{n-1}(Y)(\d\inn(Y)\Theta+\inn(Y)\d\Theta)
 \\ &=&
 \d\Lie^{n-1}(Y)\inn(Y)\Theta+\Lie^{n-1}(Y)\inn(Y)\d\Theta
 \\ &=&
 \d\Lie^{n-1}(Y)\inn(Y)\Theta-\d\xi_Y=
 \d (\Lie^{n-1}(Y)\inn(Y)\Theta-\xi_Y) \ .
 \eeann
 Hence we can write $\xi_Y=\Lie^{n-1}(Y)\inn(Y)\Theta-\zeta_Y$.
 \een
 \end{proof}



\subsection*{Acknowledgments}

{\small We acknowledge the financial support of the 
{\sl Ministerio de Ciencia e Innovaci\'on} (Spain), projects
MTM2014--54855--P,
MTM2015-69124--REDT, and of
{\sl Generalitat de Catalunya}, project 2014-SGR-634.


\end{document}